\documentclass[table]{ccjnl}
\usepackage{lipsum,amsmath}
\usepackage{cuted}
\usepackage{bm}
\usepackage{color}
\graphicspath{{figures/}}

\title{Time-Critical Tasks Implementation in MEC based Multi-Robot Cooperation Systems}

\author{Rui Yin\inst{2,*}, Yineng Shen\inst{1}, Huawei Zhu\inst{2}, Xianfu Chen\inst{3}, and Celimuge Wu\inst{4}\corinfo{yinrui@zucc.edu.cn}.
}

\receiveddate{May. 27, 2021}
\reviseddate{May. 27, 2021}
\Editor{Bei Liu}

\address[1]{School of Information and Electrical Engineering, Zhejiang University, Hangzhou 310027, China.}
\address[2]{School of Information and Electrical Engineering, Zhejiang University City College, Hangzhou 310015, China.}
\address[3]{The VTT Technical Research Centre of Finland, Oulu, 90570, Finland.}
\address[4]{Graduate School of Informatics and Engineering, the University of Electro-Communications, Tokyo, 182-8585, Japan.}

\begin{document}

\maketitle

\begin{abstract}
\label{abstract}
\emph{Mobile edge computing} (MEC) deployment in a \emph{multi-robot cooperation} (MRC) system is an effective way to accomplish the tasks in terms of energy consumption and implementation latency.
However, the computation and communication resources need to be considered jointly to fully exploit the advantages brought by the MEC technology. In this paper, the scenario where multi robots cooperate to accomplish the time-critical tasks is studied, where an intelligent \emph{master robot} (MR) acts as an edge server to provide services to multiple \emph{slave robots} (SRs) and the SRs are responsible for the environment sensing and data collection. To save energy and prolong the function time of the system, two schemes are proposed to optimize the computation and communication resources, respectively.
In the first scheme, the energy consumption of SRs is minimized and balanced while
guaranteeing that the tasks are accomplished under a time constraint. In the second scheme,
not only the energy consumption, but also the remaining energies of the SRs are considered to
enhance the robustness of the system. Through the analysis and numerical simulations, we demonstrate that even though the first policy may guarantee the minimization on the total SRs' energy consumption, the function time of MRC system by the second scheme is longer than that by the first one.
\end{abstract}
\section{Introduction}\label{Introduction}
With the development of intelligent technologies, the individual robot becomes powerful. As a consequence, more and more autonomous robots are employed to replace the human forces to provide various services, such as floor mopping, surveillance by \emph{unmanned aerial vehicles} (UAVs), surgery by surgical robots, and manufacture by mechanical arms \cite{2020-liu}. In contrast, the amount of latency-sensitive and computation-intensive tasks implemented by robots has greatly increased. For instance, in an effort to slow the spread of COVID-19, UAVs are required to monitor social distance between people \cite{20-chamola}, and the search-and-rescue tasks are instantly executed after the nature disaster. To address these issues, \emph{multi-robot cooperation} (MRC) systems have received extensive attention in recent years \cite{18-ismail}.

However, the battery budget, limited sensing ability, and computation capacity of the individual robot pose challenges to accomplish tasks \cite{18-bedi}. \emph{mobile edge computing} (MEC) technology has emerged as a promising solution to enable the task offloading via wireless communication, which perfectly fits the MRC structure. On the one hand, multiple robots need to work cooperatively to accomplish the jobs under a strict time constraint. On the other hand, efficient communication and task assignment are important to ensure effective cooperation among robots. Therefore, in this paper, we study a MEC-based MRC system where an intelligent \emph{master robot} (MR) acting as an edge server leads multiple \emph{slave robots} (SRs) to accomplish the time-critical and computation intensive tasks while minimizing and balancing the energy consumption of the robots.
\subsection{Related Works}
\subsubsection{Robot Cooperation Networks}
Recent years, \emph{Artificial intelligence} (AI) applications are extensively and effectively used to improve the capabilities of the robots on information processing attributed to their ability to deal with big data with high accuracy \cite{20-khan}. Nevertheless, the application of AI technology causes tremendous computation workloads \cite{17-kim}. In \cite{19-Alsamhi}, \emph{Machine learning} (ML) has been applied to improve the quality of connection and efficient data collection of robots. However, the cooperation among robots are not considered in this work.

To speed up time-critical robotic applications such as AI applications and robot navigation by utilizing \emph{Simultaneous Localization And Mapping} (SLAM) technology, there have been a number of studies on the following three fields \cite{18-ismail}: different types of robots, control architectures, and communication technologies. In the robot field, Boston Dynamics changes people's idea of what robots can do \cite{19-guizzo}. However, relying on limited resources, it is impossible for a single robot to complete large-scale, low-latency tasks. For example, massive environment sensing and surveillance is a typical application of robots. In order to accomplish accurate sensing and timely monitoring, the robots have to collect enough data, communicate with each other and compute cooperatively in an efficient way. 

MRC system takes the advantage of cooperation to overcome resource limitation of a single robot. First, the introduction of controlled mobility has opened up the applications of sparse sensing in \cite{13-Mostofi}, which can be exploited for energy-efficient and non-redundant sensing in MRC networks.
	Secondly, the most significant advantage of such systems lies in the fact that robots can share their collected information and resource management strategies with others by advanced wireless communication technologies. Thus, effective wireless communication among robots becomes one of the most important parts to facilitate the cooperation among robots in the MRC systems. In \cite{19-chen}, to achieve ultra-reliability and ultra-low latency, $p$-persistent real-time ALOHA has been suggested to be the multiple access protocol of the ad-hoc based MRC systems. 

With the improvement of wireless communication technologies, cloud robots have received particular attention \cite{15-kehoe}. In \cite{14-roazielo}, in order to reduce the computational workload, a visual SLAM robotic system has been proposed to utilize a distributed framework where the cloud server takes care of the map optimization and storage information. However, the data transmission latency required by cloud computing is intolerable. Therefore, to reduce the latency, the crux problem is to eastablish a communication paradigm which can facilitate the cooperation among robots in the MRC system. MEC emerges as one of the key technologies of 5G, has become a viable solution for edge-enabled applications \cite{19-liu}. In the next subsection, we introduce the latest development of MEC technology in detail.
\subsubsection{Mobile edge computing}
As a promising communication and computation offloading combination technology, MEC extends computation, communication, and storage resources toward the edge of a network. It has interdisciplinary nature between computation offloading and wireless communication, which plays an important role to reduce energy consumption and latency \cite{20-Hai}. Therefore, MEC is a natural choice to support the cooperation among robots, which can not only provide efficient communications among robots, but also assign tasks among robots effectively. In \cite{15-chen}, to minimize the task implementation latency and energy consumption, the authors have investigated a computing resources allocation scheme at mobile users and MEC servers. In \cite{17-wang}, by jointly considering computation offloading, resource allocation, and content caching strategy, an \emph{Alternating Direction Method of Multipliers} (ADMM)-based algorithm has been proposed to maximize the revenue of \emph{MEC system operator} (MSO). In \cite{17-you}, an energy-efficient task offloading scheme for the multiple devices has been proposed. To overcome network congestion and long latency in cloud computing, a collaborative cloud and edge computing resource management algorithm has been devised in \cite{19-ren}. 
\subsubsection{MEC based Cooperative Communication}
As an important technology to enable the cooperation among robots via the mobile edge server, it has been studied extensively to save the computational resources of the users in the conventional wireless communication systems. When multiple edge servers serve multiple users, centralized and distributed methods have been proposed to solve the optimal task offloading problem in \cite{18-you} and \cite{16-chen}, respectively. Furthermore, a joint offloading and trajectory design scheme has been studied to minimize the sum of maximum delay in a MEC-based UAV system \cite{18-hu}. These works have been considered a cooperation MEC system, but the user in such system was not a robot, so they did not consider unique characteristics of the robot system, such as the ability to interact with the environment.
\subsection{Our Contributions}
This paper elaborates on a MEC-based MRC system to accomplish the computation-intensive and time-critical tasks while minimizing the energy consumption of robots to prolong their function time. In the studied scenario, the intelligent MR is modeled as an edge server to
make a decision on the duration of environment sensing, the amount of data gathered by each SR, and the amount of data offloaded from each SR.
The mobile SRs are mainly in charge of environment sensing and data collection. Then, they can
offload a certain amount of tasks to the MR to save the energy consumption on
the data processing.
Two resource management strategies are proposed to minimize and balance the energy consumption among SRs, while guaranteeing that the tasks can be accomplished in time. The difference is that the remaining energy at the SRs is considered to ensure that the SRs can process the data locally without offloading in the second scheme. Therefore, it is more robust to the unreliable wireless communication environment than the first scheme, but may consume more energy.
The main contributions of the paper are summarized as follows.
\begin{enumerate}
	\item[$\bullet$] An implementation framework of the computation-intensive and time-critical task in the MEC-based MRC systems is proposed. The task implementation process is divided into three stages. In the first stage, the SRs work on environment sensing and data collection. In the second stage, the SRs may offload a certain amount of data to the MR and process the rest of the data locally.
	In the third stage, the MR needs to process the data and may also upload a certain amount of data to the \emph{base station} (BS).
	The MR makes decision on the time spent on sensing, data offloading and power scheme at the SRs at the beginning of the task implementation.
	\item[$\bullet$] Two resource management and offloading strategies are developed to ensure that the tasks can be done in time while minimizing the energy consumption of the robots. In the first scheme, the energy consumption of the SRs is minimized and balanced to enhance the function time of the MRC systems. In the second scheme, the remaining energy of the SRs is also considered to cope with the variable wireless communication environment. As a consequence, each SR may process its entire collected data locally without offloading to the MR. Therefore, it is more robust than the first scheme, but may cost more energy.
	\item[$\bullet$] Simulation results are present to verify the effectiveness of the proposed schemes, while the robustness of the MRC system is analyzed as well.
\end{enumerate}

The rest of this paper is organized as follows. In Section~\ref{System model}, we introduce the MEC-based MRC system model. Section~\ref{Problem} presents the framework and the corresponding problem formulation. Two resource management strategies are developed in Section~\ref{Optimal} and Section~\ref{suboptimal}, respectively. Numerical results are provided to verify the proposed schemes in Section~\ref{simulation}. The paper is concluded in Section~\ref{CONCLUSION}.
\section{System model}
\label{System model}
In the paper, the time-critical tasks implementation process in MEC-based MRC systems is studied. As shown in Figure~\ref{Fig.1}, a MR, denoted by $M$, works together with $K$ SRs, denoted by a set of $\mathcal{S}=\{s_1,s_2,...,s_{K}\}$ and one BS. The MR leads multiple SRs to accomplish the latency-sensitive tasks such as object identification and tracking jobs. Basically, these tasks are data-driven and computation-intensive, which requires the SRs to collect data first. Therefore, the SRs focus on sensing and collecting the data in the environment. On the other hand, the MR is more ``intelligent'' and has more energy and powerful computational capability to process the data. Hence, the MR sends the order on the amount of data required to be collected to each SR. Then, after collecting the data, each SR would transmit a certain amount of data wirelessly to the MR, which is also decided by the MR. Moreover, since the received data at the MR may be overwhelming, it has to decide whether to offload the data to the remote BS and how much to be offloaded. Hence, the MR acts as an edge server to balance the power consumption and computation load among SRs. After the data processing, the BS and SRs send the results to the MR which will respond to the human being and feedback the new order to SRs based on the comprehensive results from itself and the BS.
\begin{figure}[!t] \centering
	\begin{centering}
		\includegraphics[width=0.45\textwidth]{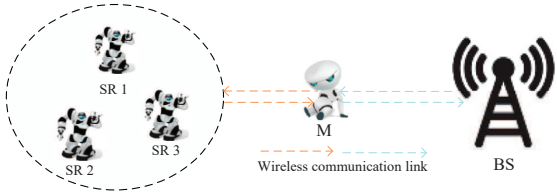}
		\caption {The model for a multi-robot system.}\label{Fig.1}
	\end{centering}
\end{figure}

As shown in Figure~\ref{Fig.2}, the process of implementing the time-critical tasks is divided into three stages. The first stage is called the sensing stage in which each SR interacts with the environment to collect data with a duration of $T_s$.
With a time limitation $T_s^{(cmp)}$, the SRs will process a certain amount of data locally and the rest of data will be offloaded to the MR. Therefore, the second stage is named as SR offloading stage. In the third stage, the MR will decide how much data to be processed by itself and how much data to be offloaded to the BS after receiving data from SRs with a duration of $T_M^{(cmp)}$, which is named as MR offloading stage.

On the one hand, efficient wireless communication among MR, SRs and the BS is the key to offload the data and accomplish the task in time when applying the MEC technology.
To guarantee reliable communication, the SRs will feedback the information on the channel state and the remaining battery power over the control channel to the MR at the beginning. Accordingly, the MR decides the amount of the data required to be collected at each SR, transmission power and spectrum allocation of each SR, and the amount of data offloaded from each SR. Then, it feedbacks the decision to the SRs over the control channel. At the beginning of the third stage,
the MR measures and estimates the channel state information between the MR and the BS first. Then, it decides the amount of the data offloaded to the BS.

On the other hand, as a group of specific task execution units, SRs may equip with different devices to interact with the environment such as the camera and Lidar. Furthermore, the SRs have much fewer battery capacities and computational capabilities than the MR and most energy would be consumed on the sensing and collecting the data in the first stage.
To maintain the function of the MRC system as long as possible, the power consumption of the SRs should be minimized and
balanced. To reach this goal, the mathematical model of the system is described in detail
in the following subsections.
\begin{figure}[h] \centering
	\begin{centering}
		\includegraphics[width=0.45\textwidth]{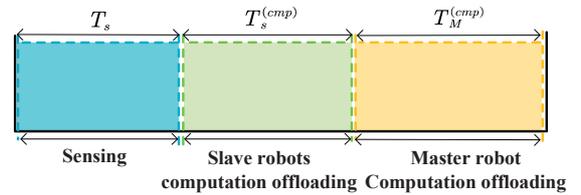}
		\caption {The time division structure for system model.}\label{Fig.2}
		\vspace{-0.3cm}
	\end{centering}
\end{figure}
\subsection{Data Collection Model}
\label{data}
The maximum duration of the sensing stage is denoted by $T_s$ in which the SRs interact with the environment to collect data.
In particular, the time spent on the sensing and collection at the $k$th SR is denoted as
$t_{k}^s$ seconds, which is a variable decided by the MR. The amount of collected data at the $k$th SR is denoted as $D_k=\frac{p_k^st_k^s}{e_k^s}$ bits, where $e_k^s$ represent the energy consumption to sense one bit at the $k$th SR \cite{03-zhu}, $p_k^s$ is the sensing power consumption at the $k$th SR, both of $e_k^s$, $p_k^s$ are constants.

\subsection{Local computation and offloading model}
\label{offloading}
Considering that the time-critical tasks are computation-intensive, the local computation at each SR may be inefficient to finish the job under a strict delay constraint. For example, the latest deep learning model, such as yolo-v4 and EfficientDet, may not be able to achieve 60 FPS with GPU V100, let alone embedded devices deployed in SRs. Thus offloading computation tasks is an efficient and necessary way to save energy at the SRs and prolong the function time of the MRC systems.

As shown in Figure~\ref{Fig.3}, define $D_k$ as the collected data which needs to be processed locally or offloaded to the MR. Define $D_k^{(off)}$ as the offloaded bits from the $k$th SR to the MR.
Then, ($D_k-D_k^{(off)}$) bits are the data bits that need to be processed locally. Moreover, let $U_k$ denote the computation capability of the $k$th SR which is the number of \emph{central processing unit} (CPU) cycles executed per second, $C_k$ as the number of CPU cycles required for computing 1-bit data at the $k$th SR.
To guarantee the latency constraint, the time spent on local computing should
be less than a threshold, which is given by
\begin{equation}\label{eq1}
	\frac{(D_k-D_k^{(off)})C_k}{U_k} \le T_{s}^{(cmp)},
\end{equation}
where $T_{s}^{(cmp)}$ is the threshold for local computing and offloading. Then we have
\begin{equation}\label{eq2}
	D_k^{(off)}\ge\max\{G_k,0\},
\end{equation}
where $G_k=D_k-\frac{U_kT_s^{(cmp)}}{C_k}$. It represents the constraint on the minimum offload data bits at the $k$th SR.

Similarly, when the MR receives the offloaded data from SRs, it also needs to offload a
certain amount of data to the BS to save both energy and time.
Define $D_M=\sum_{k\in\mathcal{S}}D_k^{(off)}$ as the total received data
from SRs, $U_M$ as the computation capacity of the MR, $C_M$ as the number of CPU cycles required for computing
1-bit of data at the MR. Then, the range of offloaded data bits from the MR to
the BS is given by
\begin{equation}\label{eq3}
	D_M^{(off)}\ge\max\{G_M,0\},
\end{equation}
where $G_M=D_M - \frac{U_MT_{s}^{(cmp)}}{C_M}$ is the minimum offload data bits at the MR.
\subsection{Transmission Model}
\label{transmission}
To offload the data to the MR from the SRs in the second stage and to the BS from the MR
in the third stage, wireless communication is applied.
According to Shannon formula, the transmission data rate from the $k$th SR
to the MR is given by
\begin{equation}\label{eq4}
	r_k = B_k\log_2(1+\frac{p_kh_k}{N_k}),
\end{equation}
where $B_k$ is the channel bandwidth allocated to the $k$th SR, $p_k$ is the transmission power consumed at SR $k$, $h_k$ denotes the mean channel gain on the wireless channel between SR $k$ and the MR, $N_k$ is the variance of complex white Gaussian channel noise experienced by SR $k$.
Based on \eqref{eq4}, the time spent on offloading data from the $k$th SR to the MR can be written as
\begin{equation}\label{eq5}
	t_k^{(off)}=\frac{D_k^{(off)}}{r_k}.
\end{equation}

It is noteworthy that the channels between the SRs and MR experience both large fading due to the path loss
and shadowing and fast fading due to the reflection and diffraction. It is given by
\begin{equation}\label{eq6}
	 h_k=15+a\log_{10}(d),
\end{equation}
where $d$ is the distance between the SR and MR, $a$ is a constant. When the system uses the licensed band, $a=3.75$.

Similarly, when the MR offloads the data to the BS, the transmission rate is given by
\begin{equation}\label{eq7}
	r_M = B_M\log_2(1+\frac{p_Mh_M}{N_1}),
\end{equation}
where $B_M$ is the bandwidth allocated to the MR, $p_M$ is
transmission power allocation at the MR, $h_M$ is the mean channel gain between the
MR and the BS, $N_1$ is the noise power, which is fixed.
According to \eqref{eq7}, the time spent on the transmission at the MR is given by
\begin{equation}\label{eq8}
	t_M^{(off)}=\frac{D_M^{(off)}}{r_M}.
\end{equation}
\begin{figure}[!t] \centering
	\begin{centering}
		\includegraphics[width=0.45\textwidth]{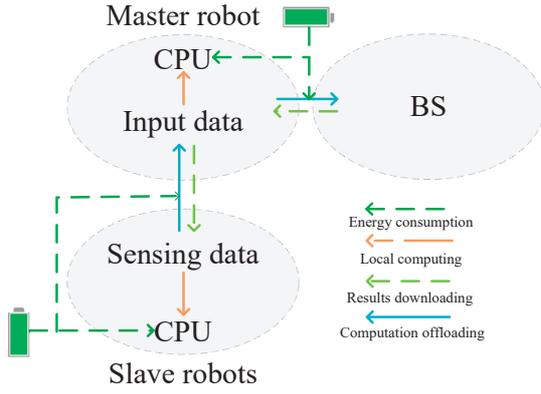}
		\caption {Data and energy flows of the system.}\label{Fig.3}
	\end{centering}
\end{figure}
\subsection{Energy Consumption Model}
\label{energy}
As shown in Figure~\ref{Fig.3}, the total energy consumption in the system can be divided into four parts, which includs the power consumption on sensing, local computing, circuit and communication. Then, according to Section~\ref{data}, the energy consumption on sensing is given by
\begin{equation}\label{eq9}
	E_k^s=p_k^st_k^s.
\end{equation}

During the local computation and data offloading at the SRs, the energy consumption on local computation is denoted by $E_{k,c}$.  Let $p_k^{(cmp)}$ be the power consumption per CPU computing cycle at the $k$th SR. Based on Section~\ref{offloading}, the data bits which need to be processed locally at the $k$th SR is $(D_k-D_k^{(off)})$.
Accordingly, the energy consumption on local computation at the $k$th SR is given by
\begin{equation}\label{eq10}
	E_k^{(cmp)}=(D_k-D_k^{off})C_kp_k^{(cmp)}.
\end{equation}

Similarly, when the MR executes the computation locally, the energy consumption is
\begin{equation}\label{eq11}
	E_{M}^{(cmp)}=(D_M-D_M^{off})C_Mp_M^{(cmp)},
\end{equation}
where $p_M^{(cmp)}$ is the power consumption per CPU cycle of computing at the MR.

The third part of the energy is consumed on data transmission at SRs and MR. We define a function $f(x)=N_k(2^{\frac{x}{B}}-1)$, and according to \eqref{eq5}, the energy consumed on data transmission at the $k$th SR is given by
\begin{equation}\label{eq12}
	E_k^{(tr)}=p_kt_k^{(off)}=\frac{t_k^{(off)}}{h_k}f(\frac{D_k^{(off)}}{t_k^{(off)}}).
\end{equation}

Similarly, the energy consumption on data transmission at the MR is given by
\begin{equation}\label{eq13}
	E_M^{(tr)}=p_Mt_M^{(off)}=\frac{t_M^{(off)}}{h_M}f(\frac{D_M^{(off)}}{t_M^{(off)}}).
\end{equation}

Based on the above analysis, the total energy consumption at the $k$th SR during the task implementation
is given by
\begin{equation}\label{eq14}
	E_k^{(tot)}=E_k^s+E_k^{(cmp)}+E_k^{(tr)}+p_k^c(t_k^s+t_k^{(off)}),
\end{equation}
where $p_k^c$ is the circuit power consumed at the $k$-th SR, which is a constant.

The total energy consumption at the MR during the task implementation is written as
\begin{equation}\label{eq15}
	E_M^{(tot)}=E_M^{(cmp)}+E_M^{(tr)}+p_M^c(t_M^{(off)}),
\end{equation}
where $p_M^c$ is the circuit power consumed at the MR, which is a constant.

It should be noted that the BS has sufficient computational capability and power, hence the computation
time spent on the feedback from the BS and the energy cost at the BS are neglected in this paper.

\section{Problem Formulation}
\label{Problem}
The aim of the proposed MEC-based MRC system is to accomplish time-critical tasks while
maintaining the function of the system as long as possible. To reach this goal,
first, we need to guarantee that the SRs collect enough data, then offload some of them to the MR in the second stage while computing the left workloads locally.
Accordingly, the optimization problem for the SRs is formulated as
\begin{equation}
	\rm{(P1) : } \hspace{-1cm}
	\begin{split}
		\min_{\{\mathbf{t}_K^s,\mathbf{D}_K^{(off)},\mathbf{t}_K^{(off)}\}} \max_{k\in\mathcal{S}}(E_k^{(tot)})
	\end{split}
\end{equation}
\vspace{-1cm}
\begin{align}
	s.t.\quad\quad
	\sum\limits_{k=1}^K D_k \geq& D, \tag{16a} \label{16a}\\
	\frac{(D_k-D_k^{(off)})}{U_k} \le& T_s^{(cmp)}, k\in\mathcal{S},\tag{16b} \label{16b}\\
	E_k^{(Re)}-E_k^{(tot)}\geq& 0, k\in\mathcal{S},\tag{16c} \label{16c} \\	
	T_s \geq t_k^s \geq& 0 ,k\in\mathcal{S},\tag{16d} \label{16d}\\
	T_s^{(cmp)} \geq t_k^{(off)} \geq&0 ,k\in\mathcal{S},\tag{16e} \label{16e}\\
	D_k \geq  &0, k\in\mathcal{S},\tag{16f} \label{16f}\\
	D_k^{(off)}\geq &0, k\in\mathcal{S},\tag{16g} \label{16g}
\end{align}
where $\mathbf{t}_K^s$ is the sensing
time vector of SRs with $\mathbf{t}_K^s=[t_i^s]_{i=1}^{K}$, $\mathbf{D}_K^{(off)}$ is the data offloading
vector of SRs with $\mathbf{D}_K^{(off)}=[D_i^{(off)}]_{i=1}^{K}$,
$\mathbf{t}_K^{(off)}$ is the time vector spent on data offloading with $\mathbf{t}_K^{(off)}=[t_i^{(off)}]_{i=1}^{K}$,
$E_k^{(Re)}$ denotes the remaining battery energy at the $k$th SR.

Constraint \eqref{16a} is to guarantee that the SRs collect enough data to accomplish
the task, \eqref{16b} is to limit the time spent on the local computation at each SR, \eqref{16c} is to guarantee that
the total power consumption at the $k$th SR is less than its remaining power,
constraints \eqref{16d} and \eqref{16e} are the time constraints for the first and second stages, respectively.

In the third stage, when the MR is working, it needs to decide how much data should be offloaded to the
BS to minimize the energy consumption while guaranteeing that the task can be finished in time.
Accordingly, the problem is written as
\begin{equation}\label{eq17}
	\rm{(P2) :}   \hspace{-1cm}
	\begin{split}
		\min_{\{t_M^{(off)},D_M^{(off)}\}}\quad E_M^{(tot)}\
	\end{split}
\end{equation}
\vspace{-0.8cm}
\begin{align}
	s.t.\quad\quad
	E_M^{(Re)}-E_M^{(tot)}\geq& 0,\tag{17a} \label{17a}\\	
	T_M^{(cmp)} \geq t_M^{(off)} \geq& 0, \tag{17b} \label{17b}\\
	\frac{(D_M-D_M^{(off)})}{U_M} \leq& T_M^{(cmp)},\tag{17c} \label{17c}\\
	D_M \geq D_M^{(off)}\geq&  0 ,\tag{17d} \label{17d}
\end{align}
where \eqref{17a} is to guarantee that the total energy consumption at the MR is less than its remaining battery energy, \eqref{17b} represents the maximum latency constraint,
\eqref{17c} denotes that the local computation time should be less
than $T_M^{(cmp)}$, and \eqref{17d} is to guarantee that the offloaded data bits are less than the total received data bits from SRs.
\section{minimizing the system energy consumption}
\label{Optimal}
In this section, we design an optimal task offloading, transmission power, and time allocation scheme for the SRs first.
\subsection{An Optimal Scheme for Solving the Problem (P1) and Problem (P2)}

According to \cite{convex_boyd}, define a variable $e$ to replace the objective function in (P1). Then, the (P1) is converted into:
\begin{equation}\label{eq18}
	\rm{(P3) :} 
	\begin{split} \hspace{-1cm}
		\min_{\{\mathbf{t}_K^s,\mathbf{D}_K^{(off)},\mathbf{t}_K^{(off)},e\}} e
	\end{split}
\end{equation}
\vspace{-0.8cm}
\begin{align}
	s.t.\quad\quad E_k^{(tot)} \leq e, k\in\mathcal{S},\tag{18a} \label{18a}\\
	(16a)-(16g).\notag
\end{align}

Then, we have the following Lemma.
\begin{lemma}
	Problem (P3) is a convex optimization problem.
\end{lemma}
\begin{proof}
	Please refer to Appendix \ref{A}.
\end{proof}
By analyzing the problem (P3), the energy consumption of the SRs in the first and second stages are coupled due to the constraints
\eqref{16c} and \eqref{18a}. To decouple these two constraints, the Lagrangian dual method is applied.  Define Lagrange multipliers, $\boldsymbol{\lambda}=[\lambda_1,...,\lambda_n]^T, \boldsymbol{\alpha}=[\alpha_1,...,\alpha_k]^T$, associated with the constraints \eqref{16c} and \eqref{18a}, respectively. Then, Lagrange dual function corresponding to \eqref{eq18} is given by:
\begin{eqnarray}\label{eq19}
	\begin{split}
		L= e+&\sum_{k=1}^K\lambda_k(E_k^{(tot)}-E_k^{(Re)})\\
		+&\sum_{k=1}^K\alpha_k(E_k^{(tot)}-e).
	\end{split} 	
\end{eqnarray}

According to \eqref{eq19}, the dual problem is defined as:
\begin{equation} \label{eq20}
	\begin{split}
		&\max_{\{\boldsymbol{\lambda}\geq0,\boldsymbol{\alpha}\geq0\}} D(\boldsymbol{\lambda},\boldsymbol{\alpha})\\
		=&\max_{\{\boldsymbol{\lambda}\geq0,\boldsymbol{\alpha}\geq0\}} \{ \min_{\{\mathbf{t}_K^s,\mathbf{D}_K^{(off)},\mathbf{t}_K^{(off)},e\}} L\},
	\end{split}
\end{equation}
which is convex on $\boldsymbol{\lambda}$ and $\boldsymbol{\alpha}$. Based on Lagrange dual method, \eqref{eq20} can be decoupled into
two problems, which are the Lagrange dual problem, $\max_{\boldsymbol{\lambda}\geq0,\boldsymbol{\alpha}\geq0}D(\boldsymbol{\lambda},\boldsymbol{\alpha})$, and the sub-prime problem:
\begin{equation} \label{eq21}
	\begin{split}
		&\min_{\{\mathbf{t}_K^s,\mathbf{D}_K^{(off)},\mathbf{t}_K^{(off)},e\}} L(\mathbf{t}_K^s,\mathbf{D}_K^{(off)},\mathbf{t}_K^{(off)},\boldsymbol{\lambda},\boldsymbol{\alpha},e)\\
		&=\min_{\{\mathbf{t}_K^s,\mathbf{D}_K^{(off)},\mathbf{t}_K^{(off)},e\}}
		((1-\sum_{k=1}^K{\alpha_k})e+\sum_{k=1}^K(\alpha_k+\lambda_k)E_k^s\\
		&+\sum_{k=1}^K(\alpha_k+\lambda_k)(E_k^{(cmp)}+E_k^{(tr)})\\
		&-\sum_{k=1}^K\lambda_kE_k^{(Re)})
	\end{split}
\end{equation}

After decoupling, the sub-problem can be decomposed into three sub-problems, named as (SP1), (SP2), and (SP3).
(SP1) is to minimize the weighted energy consumption of SRs during the sensing stage, which is given by
\vspace{-0.3cm}
\begin{equation}\label{eq22}
	\rm{(SP1):} \hspace{-1.2cm}
	\begin{split}
		\min_{\{\mathbf{t}_K^s,\mathbf{D}_K\}}	\sum_{k=1}^K(\alpha_k+\lambda_k)E_k^s
	\end{split}
\end{equation}
\vspace{-0.8cm}
\begin{align}
	s.t.\quad\quad (16a), (16d), (16f)\notag.
\end{align}

(SP2) is to minimize the weighted energy consumption of SRs during the local computation and
task offloading at the second stage, which is given by
\vspace{-0.3cm}
\begin{equation}\label{eq23}
	\rm{(SP2) :} \hspace{-0.5cm}
	\begin{split}
		\min_{\{\mathbf{D}_K^{(off)},\mathbf{t}_K^{(off)}\}} \sum_{k=1}^K(\alpha_k+\lambda_k)(E_k^{(cmp)}+E_k^{(tr)})
	\end{split}
\end{equation}
\vspace{-0.8cm}
\begin{align}
	s.t.\quad\quad (16b), (16e), (16g).\notag
\end{align}

The optimal maximum total energy consumption among SRs can be obtained by solving the following sub-problem:
\vspace{-0.3cm}
\begin{equation}\label{eq24}
	\rm{(SP3) :}  \hspace{-1.2cm}
	\begin{split}
		\min_{\{e\}}\quad (1-\sum_{k=1}^K{\alpha_k})e
	\end{split}
\end{equation}
\vspace{-1cm}
\begin{align}
	s.t.\quad\quad e\geq 0 .\tag{24a} \label{24a}
\end{align}

\begin{algorithm}[H]
	\caption{Minimizing the System Energy Consumption}
	\begin{algorithmic}[1]
		\STATE \textbf{Initialize:} Iteration index $t=1$, and system parameters.
		\STATE \textbf{While:} Problem has a feasible solution:
		\STATE \quad With the current $\lambda_k$, the Interior Point Method and BCD Method can be applied to the (SP1), (SP3) and (SP2) to obtain the globally optimal solution $\{D_{t,k},t_{t,k}^s,D_{t,k}^{(off)},t_{t,k}^{(off)}\}$;
		\STATE \quad Use \emph{Gradient Descent (GD)} Method to update paramters $\boldsymbol{\lambda}_k$;
		\STATE \quad Return to step 4 until convergence;
		\STATE \quad Compute total energy consumption $\mathbf{E}_{t,k}^{(tot)}$;
		\STATE \quad Update the remaining energy of each slave robots, i.e. $E_{t,k}^{Re}=E_{t,k}^{Re}-E_{t,k}^{(tot)}$;
		\STATE \quad Solve problem (P2) by BCD Method and get the optimal solution $\{D_{t,M}^{(off)},t_{t,M}^{(off)}\}$;
		\STATE \quad Compute energy consumption $E_{t,M}^{(tot)}$ at the MR;
		\STATE \quad Update the remaining energy of the MR, i.e., $E_{t,M}^{Re}=E_{t,M}^{Re}-E_{t,M}^{(tot)}$;
		\STATE \quad Update $t=t+1$;
	\end{algorithmic}
\end{algorithm}\label{alg1}
With Lagrange multipliers, $\boldsymbol{\lambda}$, $\boldsymbol{\alpha}$, (SP1), (SP2) and (SP3) are all convex optimization problems. In particular, (SP1) and (SP3) are linear programming problems, which can be solved by the Interior Point Method.
The \emph{Block Coordinate Descent} (BCD) \cite{BCD} optimization technique can be used to obtain the optimal solution of (SP2).

After achieving the optimal solution of the sub-prime problem, the Lagrange dual problem is a linear programming problem, which is formulated as:
\begin{equation}\label{eq25}
	\begin{split}
		&\max_{\{\boldsymbol{\lambda},\boldsymbol{\alpha}\}}((1-\sum_{k=1}^K\alpha_k)e^*+\sum_{k=1}^K(\alpha_k+\lambda_k)E_k^{s*}\\
		&+\sum_{k=1}^K(\alpha_k+\lambda_k)(E_k^{(cmp)*}+E_k^{(tr)*})\\
		&-\sum_{k=1}^K\lambda_kE_k^{(Re)})
	\end{split}
\end{equation}
\vspace{-1cm}
\begin{align}
	s.t. \quad \lambda_k \geq0 ,k\in\mathcal{S} ,\tag{25a} \label{25a}\\
	\alpha_k \geq 0 ,k\in\mathcal{S} ,\tag{25b} \label{25b}
\end{align}
where $e^*,E_k^{s*},E_k^{(cmp)*},E_k^{(tr)*}$ are the achieved optimal solutions of (SP1), (SP2) and (SP3). The updating rule in the following algorithm $1$ can be applied to derive the optimal $\boldsymbol{\lambda},\boldsymbol{\alpha}$.

Based on the Interior Point Method, the solution of (P2) can be achieved and the optimal task offloading and resource management scheme can be derived for the MR since it is a convex optimization problem. According to the above analysis, the optimal joint computation and communication resource scheme is concluded in Algorithm 1.

\section{The robust task offloading and resource management scheme}
\label{suboptimal}

In problem (P1), only the fairness on the energy consumption among SRs is considered, which does not take the remaining energy of the SRs into account. In this section, the task offloading and resource management policies of the SRs at the first (sensing) and second (SR offloading) stages are studied separately, which considers not only the energy consumption,
but also the remaining energy of the SRs. First, the optimization problem for the first stage is formulated as
\begin{equation}\label{eq26}
	\rm{(P4) :} \hspace{-1.2cm}
	\begin{split}
		\min_{\{\mathbf{t}_{k}^s\}}\sum_{k=1}^K (\beta_k E_k^s)
	\end{split}
\end{equation}
\vspace{-0.8cm}
\begin{align}
	s.t.\quad\quad
	E_k^{Re}-D_kC_kp_k^{(cmp)}-E_k^s\geq& 0, k\in\mathcal{S}\tag{26a} \label{26a},\\
	(16a),(16d),(16f), \notag
\end{align}
where \eqref{26a} ensures that the remaining battery energy is enough to implement data collection and local computation without offloading. $\beta_k$ is the weighted factor associated with the $k$th SR, which is related to the fairness on the remaining energy. It is given by
\begin{equation}\label{eq27} 	
\beta_k = \frac{\min(E_k^{Re})}{E_k^{Re}},k\in\mathcal{S}.
\end{equation}

Intuitively, when the remaining energy of the SR is less, its corresponding weight factor will be bigger based on \eqref{eq27}.
Then, based on the definition of the objective function in (P4), less amount of data collection will be assigned to
the SR.

It is noteworthy that the constraint \eqref{26a} is different from the constraint \eqref{16c} in problem (P1). The constraint
\eqref{16c} is to ensure that the remaining power at each SR is enough to accomplish the sensing, local computation and offloading. However, in certain scenarios, the wireless channel is not good enough to support the data transmission between the SRs and MR, i.e. the SRs enter the dead zone. We have to make sure the SRs can handle the collected data by themselves under this condition. Therefore, constraint \eqref{16c} is replaced with \eqref{26a} in (P4).

Then, in the SR offloading stage, we need to minimize and balance the energy consumption among SRs. Therefore, the problem is formulated as
\begin{equation}
	\rm{(P5) :} \hspace{-1cm}
	\begin{split}
		\min_{\{\mathbf{t}_{k}^{(off)},\mathbf{D}_k^{(off)}\}} \sum_{k=1}^K(E_k^{(cmp)}+E_k^{(tr)})
	\end{split}
\end{equation}
\vspace{-0.8cm}
\begin{align}
	s.t.\quad\quad
	(16b),(16e),(16g).\notag
\end{align}

Problem (P4) is a linear programming problem, which can be easily solved by the Interior Point Method. Moreover, problem (P5) is a convex optimization problem. We apply the Lagrangian multiplier method \cite{convex_boyd} to derive the close-form solution of (P5).

The partial Lagrange function of (P5) is defined as
\begin{equation}\label{eq29}
	\begin{split}
		L=&\sum_{k=1}^K(\frac{t_{k}^{(off)}}{h_k}f(\frac{D_k^{(off)}}{t_{k}^{(off)}})+(D_k-D_k^{(off)})C_kp_k^{(cmp)})\\
		&+\sum_{k=1}^K\phi_k(t_k^{(off)}-T_s^{(cmp)}),
	\end{split}
\end{equation}
where $\boldsymbol{\phi_k}$, $\boldsymbol{\phi_k}\ge0$, is the Lagrange multiplier associated with the constraint \eqref{16e}. In order to facilitate the subsequent analysis, we define $g(x)=f(x)-xf^{'}(x)$.

Then \emph{Karush-Kuhn-Tucker} (KKT) conditions for (P5) \cite{convex_boyd} leads to the following necessary and sufficient conditions:
\begin{equation}\label{eq30} 		
	\begin{split}
		\frac{\partial L}{\partial D_k^{(off)*}} =& \\
		\frac{f^{'}(\frac{D_k^{(off)*}}{t_{k}^{(off)*}})}{h_k}&-C_kp_k^{(cmp)}	\begin{cases}
			\textgreater 0, & D_k^{(off)*}=G_k \\
			=0, &D_k^{(off)*} \in (G_k,D_k) \\
			\textless0, &D_k^{(off)*} = D_k \\
		\end{cases},
	\end{split}
\end{equation}
\begin{equation}\label{eq31} 	
	{\frac{\partial L}{\partial t_{k}^{(off)*}} = \frac{g(\frac{D_k^{(off) *}}{t_{k}^{(off)*}})}{h_k}+\phi_k^*}	\begin{cases}
		\textgreater 0, & t_{k}^{(off)*}=0 \\
		=0, & t_{k}^{(off)*}>0 \\
	\end{cases},
\end{equation}
\begin{equation}\label{eq32}
	{{t_k^{(off)*}}-T_s^{(cmp)} \leq 0,k\in\mathcal{S}},
\end{equation}
\begin{equation}\label{eq33}
	{\phi_k^*({t_k^{(off)*}}-T_s^{(cmp)})=0,k\in\mathcal{S}},
\end{equation}
where $\{D_{k}^{(off)*},t_{k}^{(off)*}\}$ denote the optimal solution for (P5).

Based on the knowledge that $D_k^{(off)*} \in (G_k,D_k)$ and $t_{k}^{(off)*}>0$, we derive the \eqref{eq34} as below from \eqref{eq30} and \eqref{eq31}:
\begin{equation}\label{eq34}
	{\frac{D_k^{(off)*}}{t_{k}^{(off)*}}=f^{'-1}(C_kp_k^{(cmp)}h_k)=g^{-1}(-h_k\phi_k^*)}.
\end{equation}

By solving \eqref{eq34} and the definition of $f(x)$ and $g(x)$, we derive a threshold function $\gamma_k$ as follows
\begin{equation}\label{eq35} 	
	\begin{split}
		\gamma_k&=g(f^{'-1}(C_kp_k^{(cmp)}h_k))\\
		&=	
		\begin{cases}
			&\frac{N_k}{h_k}(O_klnO_k-O_k+1),  ~~O_k\geq1,\\
			&0,  ~~~~~~~~~~~~~~~~~~~~~~~~~~~~O_k<1,
		\end{cases}
	\end{split}
\end{equation}
the equation \eqref{eq35} is derived and its properties are proved in Appendix \ref{B}.

\subsection{Offloading Policy}
In this subsection, an threshold-based policy for solving the (P5) is described in detail. First, a priority indicator $O_k$ is defined as follow
\begin{eqnarray}\label{eq36} 	
	\small 	O_k=\frac{B_kh_kC_kp_k^{(cmp)}}{N_k\ln2}, k\in\mathcal{S}.
\end{eqnarray}

Note that $O_k$ depends on the corresponding parameters quantifying uplink channel $h_k$, channel bandwidth $B_k$, and robots local computation capacity $C_kp_k^{(cmp)}$. Obviously, $O_k$ is a monotonically increasing function of $B_k,h_k,C_k,$ and $p_k^{(cmp)}$. Next, the relationship between the optimal time fraction, $t_{k}^{(off)*}$, and offloading data size, $D_{k}^{(off)*}$, for the $k$th SR with the corresponding offloading priority indicator is shown in the following Theorem.
\begin{theorem}(offloading strategy for problem (P5))
	\begin{enumerate}
		\item If $O_k\leq1$, $D_k^{(off) *}=t_k^{(off)*}=G_k, k\in\mathcal{S},$\\
		specifically, if $G_k\leq0$, $D_k^{(off) *}=t_k^{(off)*}=0$.
		\item If $O_k\textgreater1$, $D_k^{(off) *}
		\begin{cases}
			=G_k,~~~~~~~ \gamma_k\textless\phi_k^*, k\in\mathcal{S}\\
			\in[G_k,D_k], 	~\gamma_k=\phi_k^*, k\in\mathcal{S}\\
			=D_k,~~~~~~~ \gamma_k\textgreater\phi_k^*, k\in\mathcal{S}
		\end{cases}$,
	\end{enumerate}
	and time fraction $t_k^{(off)*}$ is expressed as follow
	\begin{eqnarray} \label{eq37}
		\small 	t_k^{(off)*}=\frac{B_k[W_0(\frac{-h_k\gamma+N_k}{-N_ke})+1]}{ln2}D_k^{(off) *}.
	\end{eqnarray}
\end{theorem}
\begin{proof}
	Please refer to Appendix \ref{C}.
\end{proof}
\begin{algorithm}[htbp]
	\caption{The Robust Scheme}
	\begin{algorithmic}[1]
		\STATE \textbf{Initialize:} Iteration index $t=1$, and system parameters. Set  $\phi_{low}=0,\phi_{high}=max(\gamma_k),\lambda_{M}^{(low)}=0,\lambda_{M}^{(high)}=\gamma_M$, where $t_{k}^{(low)},t_{k}^{(high)},t_{M}^{(low)},t_{M}^{(high)}$ are the allocated fractions for the cases $\phi_{low},\phi_{high},\lambda_{M}^{(low)},\lambda_{M}^{(high)}$. Then according to Theorem 1, compute the paramter $O_k,O_M$.
		\STATE \textbf{While:} Problem has a feasible solution
		\STATE \quad Solve problem (P4) by the Interior Point Method, then obtain $\{D_{t,k},t_{t,k}^s\}$;
		\STATE \quad Accroding to Theorem 1, while $t_{k}^{(low)} \neq T_s^{(cmp)}$ and $t_{k}^{(high)} \neq T_s^{(cmp)}$, use Bisection search to each SR, update $\phi_{low},\phi_{high}$ until convergence.
		\STATE \quad Update $E_{t,k}^{Re}=E_{t,k}^{Re}-E_{t,k}^{(tot)}$;
		\STATE \quad Accroding to Theorem 1, while $T_{M}^{(low)} \neq T_M^{(cmp)}$ and $T_{M}^{(high)} \neq T_M^{(cmp)}$, use Bisection search to each SR, update $\lambda_{M}^{(low)},\lambda_{M}^{(high)}$ until convergence.
		\STATE \quad Update $E_t^{Re}=E_t^{Re}-E_{t,M}^{(tot)}$;
		\STATE \quad Updata $t=t+1$;
	\end{algorithmic}
	\label{algorithm1}
\end{algorithm}
According to Theorem 1, a similar scheme can be used to allocate optimal time fraction and offloading data size for MR. Then, the whole process of solving (P2), (P4), and (P5) is summarized in Algorithm 2.

\begin{remark}
	Theorem 1 reveals that the offloading decisions for SRs have a threshold-based structure. Since $\gamma_k=\phi_k^*$ rarely occurs in practice, the optimal offloading decisions basically are a binary structure \cite{17-mao}, which means to offload all data or not offload all data, for each cooperative SR.
\end{remark}
\begin{remark}
	Algorithm 2 has low computation complexity. Specifically, given a solution accuracy $\epsilon \textgreater 0$, the iteration complexity for one-dimensional search is given as $\mathcal{O}(\log(1/\epsilon))$. And the computation complexity of inner loop which use the Interior Point Method is $\mathcal{O}(K^{3.5}\log(1/\epsilon))$. Since it has $N$ times iteration, the complexity of Algorithm 2 is $\mathcal{O}(N((K^{3.5}+2)\log(1/\epsilon)))$. Furthermore, the Algorithm 1's computation complexity is analyzed as follow. As the BCD iterations are in the complexity of the order $\log(1/\epsilon)$, and the complexity of  Method is $\mathcal{O}(K ^3\log(1/\epsilon))$. Then the total complexity of the Algorithm 1 is $\mathcal{O}(N(2K^{3.5}+K^3+2)\log(1/\epsilon))$. Obviously, Algorithm 2 has much lower complexity.
\end{remark}
\begin{table*}[htbp]
	\centering
	\caption{Simulation Parameters}
	\begin{tabular}{|c|c|c|c|}\hline
		Parameters & Value &Parameters & Value\\\hline	
		Number of SR & 3 &Number of MR & 1 \\\hline
		Path loss model:& &The variance of complex & \\
		Licensed: $a=3.75$ & $15+alog_{10}(d)$&white Gaussian channel noise: & $\in[-86,-96]$\\
		Unlicensed: $a=5$ & $d\in[1,100]$ m &$N_k$ & dBm \\\hline
		Sensing time limitation: $T_s$ & $800$ ms &Offloading time limitation: $T_{s}^{(cmp)}$ &$40$ ms  \\\hline
		Bandwidth: $B_k,B_M$ &  $10$ MHz & Offloading time limitation: $T_{M}^{(cmp)}$ & $10$ ms  \\\hline
		The energy consumption & $\in[18,22]$ &The sensing power: $p_k^s$ & $\in [0.05,0.13]$\\
		needed to sense a bit: $e_k^s$&  nJ/bit & & W\\\hline
		The required number of  & $\in [100,1000]$ &The required energy &$ \in [100,500]$ \\
		CPU cycles per bit: $C_k$ & cycles/bit &consumption per cycle: $p_k^{(cmp)}$ &$\times 10^{-11}$ J/cycle  \\\hline
		The computation & $[100,...,800]\times10^6$ &The required number &$100$ \\
		capability: $U_k$ & cycles/s &of CPU cycles per bit: $C_M$ &cycles/bit \\\hline
		The required energy & $ 1000\times10^{-11}$&The initial battery & $[3,2.5,0.5]$ J \\
		consumption per cycle: $p_M^{(cmp)}$ & J/cycle& power of each SR&   \\\hline
		The initial battery power of MR & $10$ J &Low power mode threshold & $0.2$ J \\\hline
		Packet Size & 5000000 bits &The computation capability: $U_M$ & $2.4\times10^9$ cycles/s \\\hline
	\end{tabular}
	\label{tab1}
\end{table*}
\begin{figure*}[h]
	\centering
	\subfloat[The sensing data size allocation and data offloading . ]{
		\includegraphics[width=0.5\textwidth]{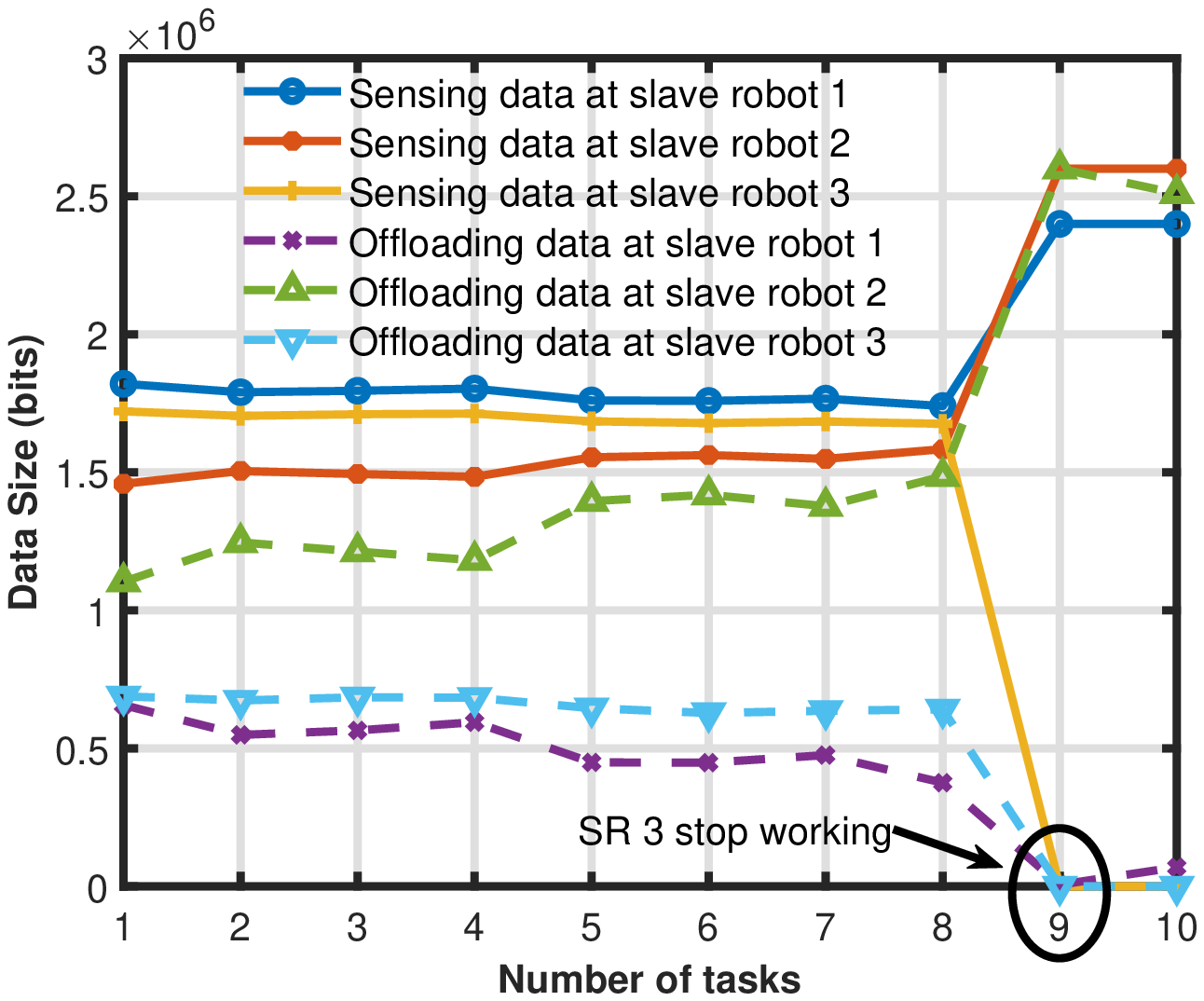}
	}
	\subfloat[The remaining battery energy corresponding to the three SRs.]{
		\includegraphics[width=0.5\textwidth]{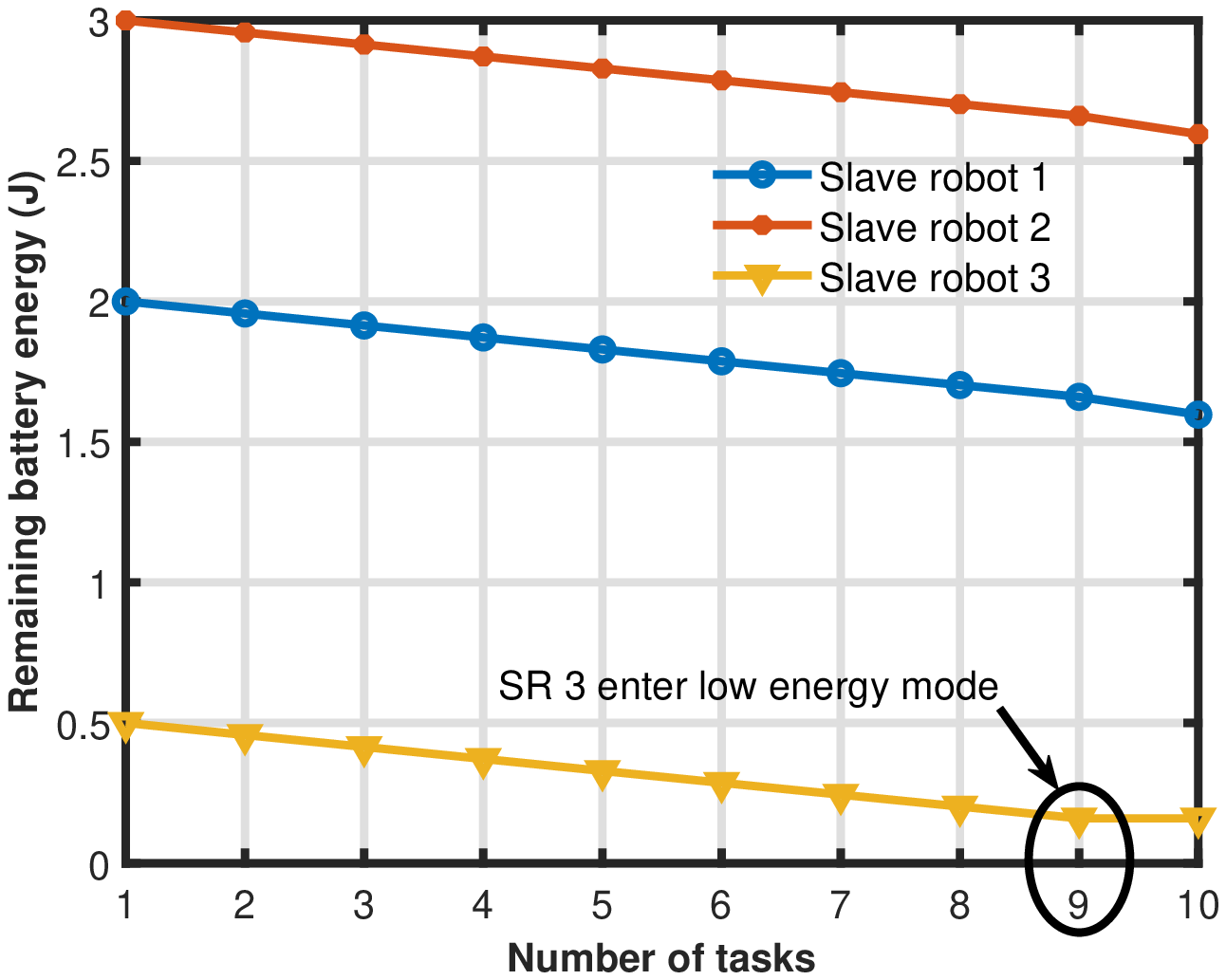}
	}
	\caption{Performance evaluation on MRC-OP.}
	\label{fig4}
\end{figure*}
\begin{figure}[t]
	\centering
	\includegraphics[width=0.5\textwidth]{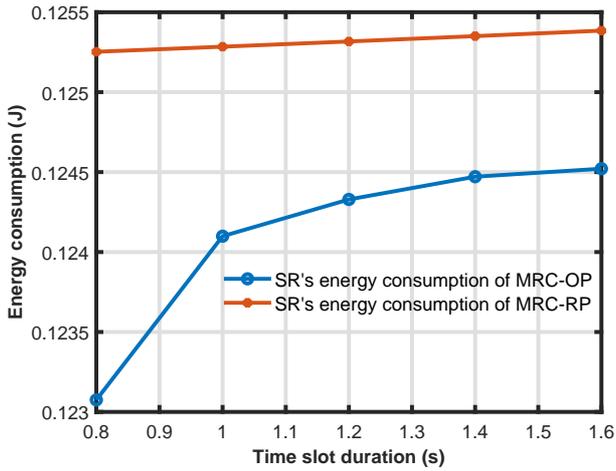}
	\caption{The impact of the constriant on $T_s$.}
	\label{fig5}
\end{figure}
\begin{figure}[t]
	\centering
	\includegraphics[width=0.5\textwidth]{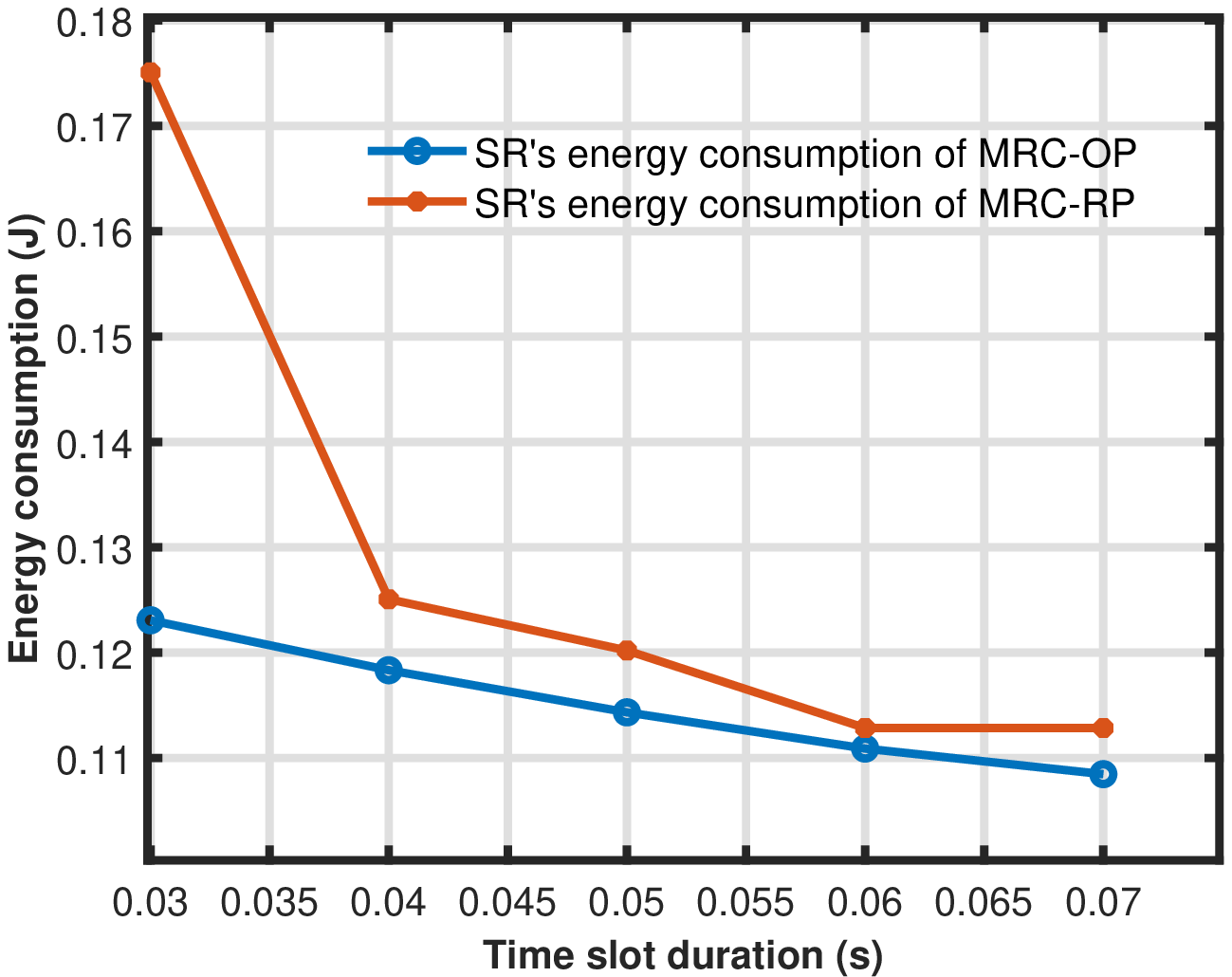}
	\caption{The impact of the constriant on $T_s^{(cmp)}$.}
	\label{fig6}
\end{figure}
\section{Numerical Results}
\label{simulation}

In this section, we provide numerical results to validate the effectiveness of Algorithm $1$ and Algorithm $2$. We refer to Algorithm 1 as ``MRC-OP", and Algorithm 2 as ``MRC-RP". There is an MR leading three SRs in simulation scenario. The key parameters used in the simulations are listed in Table \ref{tab1}.
\begin{figure*}[t]
	\centering
	\subfloat[The sensing data size allocation and data offloading.]{
		\includegraphics[width=0.5\textwidth]{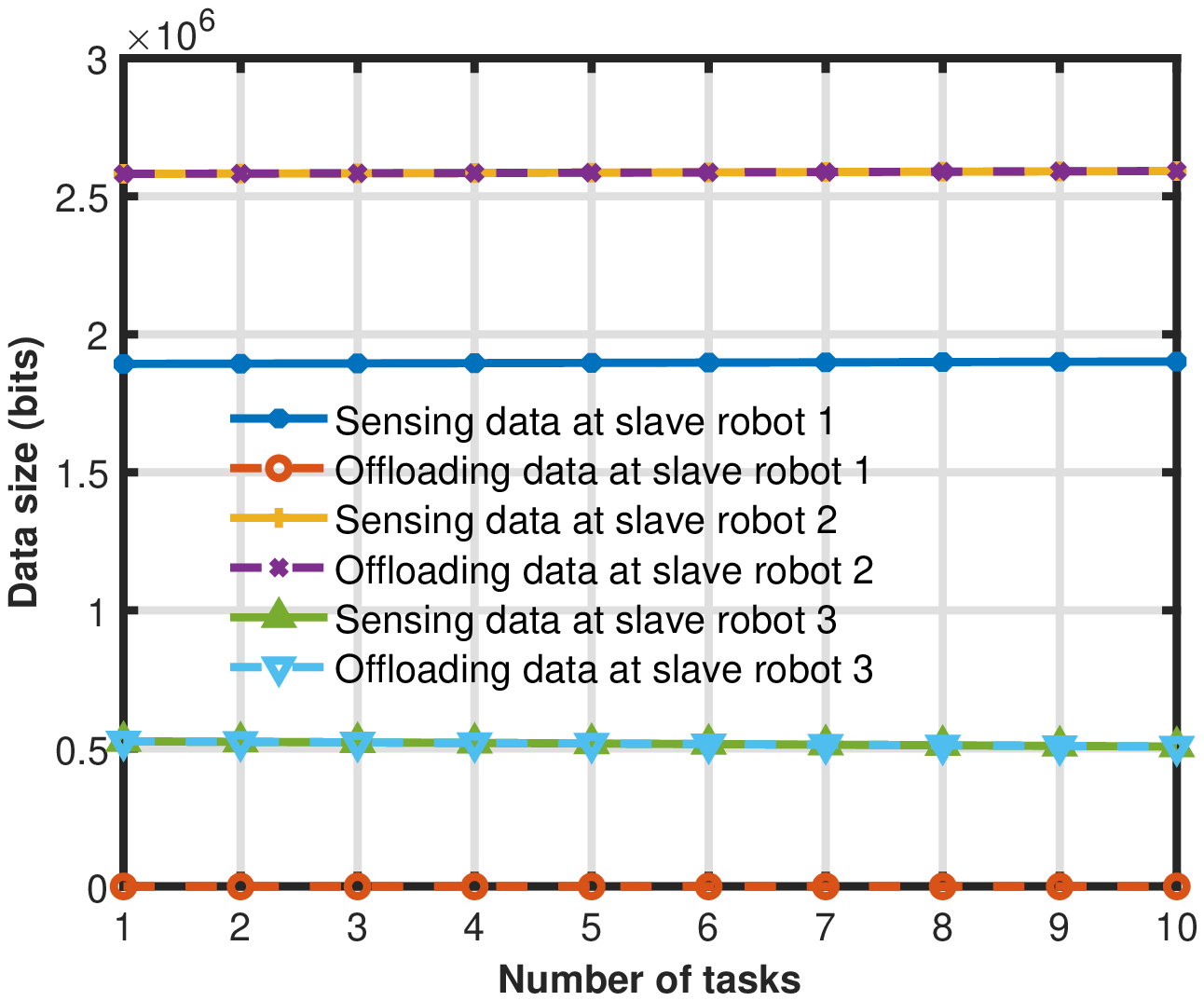}
	}
	\subfloat[The comparsion of remaining battery energy by MRC-OP and MRC-RP schemes.]{
		\includegraphics[width=0.5\textwidth]{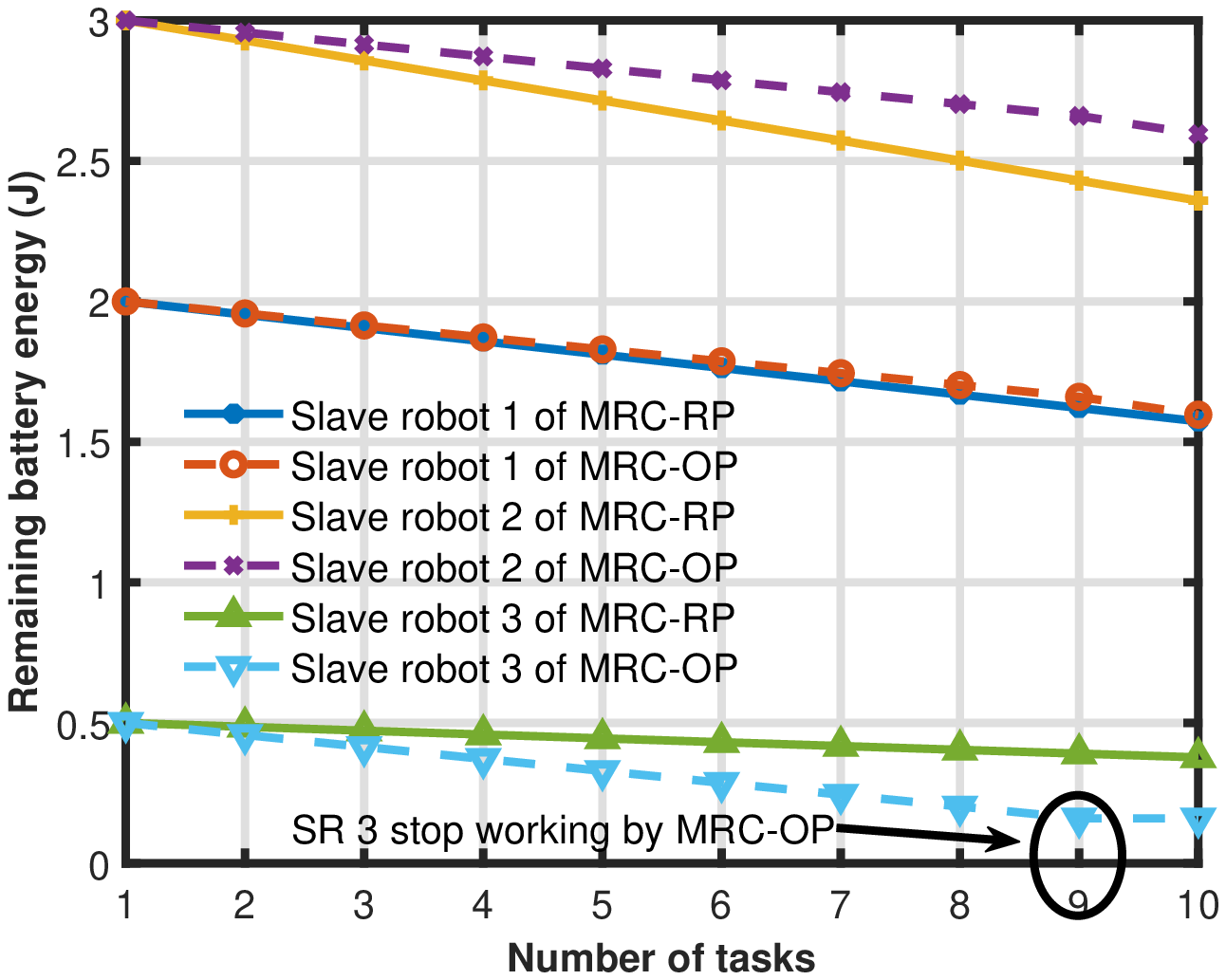}
	}
	\caption{Performance evaluation on MRC-RP.}
	\label{fig7}
\end{figure*}
\begin{figure*}[t]
	\centering
	\subfloat[The comparison of energy consumption by MRC-OP, MRC-RP and GOP.]{
		\includegraphics[width=0.5\textwidth]{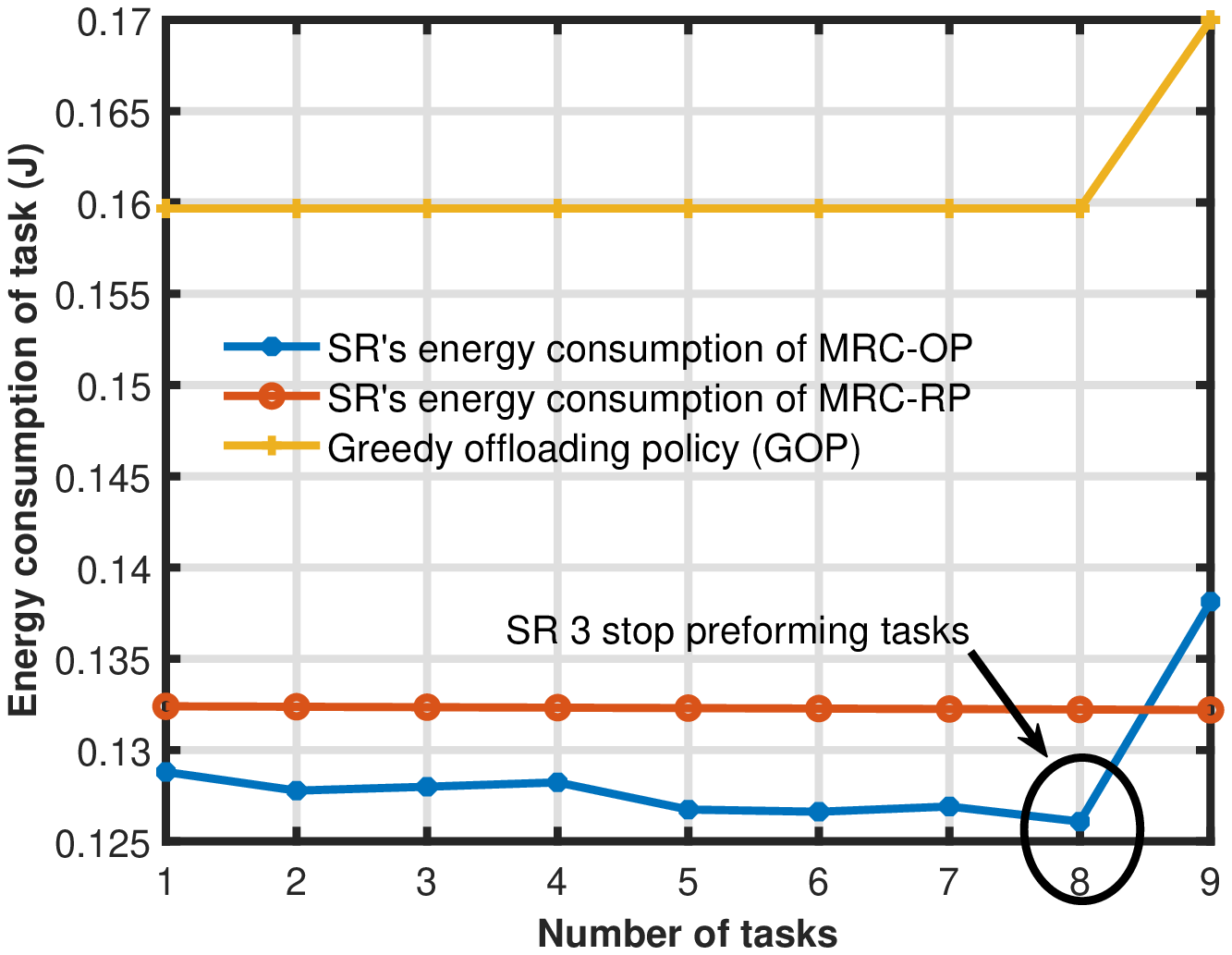}
	}
	\subfloat[The comparison of energy consumption after changing channel gain.]{
		\includegraphics[width=0.5\textwidth]{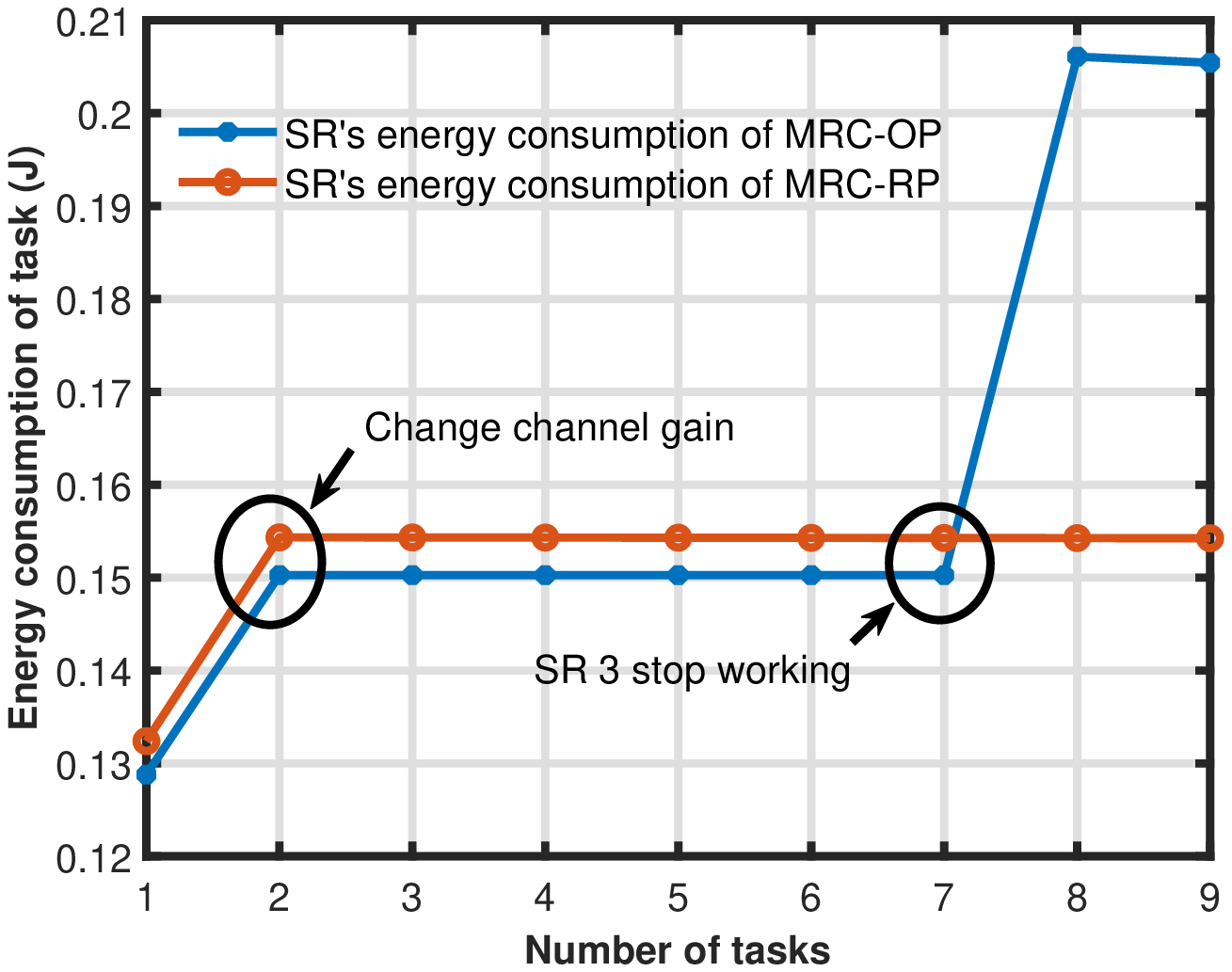}
	}
	\caption{The robustness of MRC-RP.}
	\label{fig8}
\end{figure*}

\subsection{MRC-OP Algorithm Performance}
Figure~\ref{fig4} (a) shows the sensing data size allocation and offloading data size versus the number of accomplished tasks. In this simulation scenario, the SR$2$ and SR$3$ experience good channel conditions. Since the energy consumption per bit on transmission is less than the local computation energy consumption per bit, the SR$2$ and SR$3$ with good channel conditions tend to offload more data. On the other hand, due to the bad channel condition and time limit on local computating and task offloading, $T_s^{(cmp)}$, SR$1$ is not willing to offload. Moreover, after SR$3$ runs out of the battery energy as indicated by the circle in the figure, the MRC-OP scheme has to reallocate the sensing data size and time resource. Figure~\ref{fig4} (b) demonstrates the remaining energy at each SR versus the number of implementing tasks. We can observe, as the number of accomplished tasks increases, more energy would be consumed and the remaining energy at the SRs decreases as well. It is noteworthy that when a robot's remaining energy falls below a threshold, the robot has to enter a low power mode and stop working, as indicating by the circle in Figure~\ref{fig4} (b).

\subsection{Impact of the Latency Constraints}
The influence of the latency constraints on $T_s$ and $T_s^{(cmp)}$ are presented in Figure~\ref{fig5} and Figure~\ref{fig6}, respectively. It is shown that, the proposed MRC-OP scheme outperforms the MRC-RP in terms of energy consumption as the time constraint on sensing duration getting weak. Specifically, when $T_s$ increases, the total energy consumption of SRs by MRC-OP is lower than those by MRC-RP, indicating the effectiveness on energy consumption minimization of MRC-OP. However, the cost on the energy consumption reduction by MRC-OP is the imbalance on the remaining energies among SRs, which will be demonstrated in next subsection.
Next, Figure~\ref{fig6} plots the total energy consumption of SRs versus time limitation on $T_s^{(cmp)}$. As $T_s^{(cmp)}$ goes up, the SRs may spend more time on the offloading. Therefore the total energy consumption on SRs would decrease. 
The above results verify that the proposed  schemes are highly capable of dealing with latency-critical tasks and reducing the energy consumption by fully taking the benefits of MEC technology.

\subsection{MRC-RP Algorithm Performance}
First, we demonstrate the sensing task allocation to the SRs and data offloading at the SRs by MRC-RP in Figure~\ref{fig7}(a).  From the figure, we can observe that the offloading strategy in MRC-RP has a binary structure, where SR1 does not offload, and SR2 and SR3 will offload all their  data to the MR. Therefore, Theorem 1 is verified.
For comparison, we demonstrate the remaining energy of SRs by MRC-OP and MRC-RP, respectively, in Figure~\ref{fig7} (b).
From the figure, we can observe that the SR$3$ can not continue to work after eight tasks are accomplished when MRC-OP is applied. However, SR$3$ can still work by MRC-RP after the accomplishment of eight tasks. Therefore, compared with MEC-OP, MRC-RP allows the individual SR to work longer than those by MRC-OP scheme since the remaining energy of SRs is considered in MEC-RP to define the fairness coefficient. 

\subsection{Comparison on the Robustness}
Figure~\ref{fig8} (a) depicts the total energy consumption of SRs with different numbers of accomplished tasks. For performance comparison, a baseline, \emph{greedy offloading policy} (GOP), is considered, where the SRs would offload as many task bits as the time allows. As we can observer that when all SRs work, the total energy consumption of SRs by MRC-OP is lower than that by MRC-RP and baseline. However, the SR$3$ stop working by MRC-OP and GOP earlier than the MRC-RP scheme. As a consequence, the total energy consumption of SRs by MRC-OP and GOP would be much greater than that by MRC-RP. This is because that the remaining robots have to do more jobs after SR$3$ stops working when MRC-OP or GOP scheme is employed. On the other hand, the MRC-RP scheme takes into account the remaining energy of SRs in advance. Therefore, it can ensure that each SR functions longer than those by MRC-OP and GOP scheme.

Moreover, Figure~\ref{fig8} (b) presents how the system energy consumption changes when the channel gains experienced by the SRs decrease. We set the channel gain as $10^{-6}$ when implementing the first task. Afterwards, we channge the channel gain to $10^{-8}$, by which SR$3$ can only work for 7 tasks by MRC-OP. But, it still works when MRC-RP scheme is used, which clearly illustrates the MRC-RP scheme is more robust to against a variable wireless channel environment.

\section{CONCLUSION}
\label{CONCLUSION}
In this paper, we investigated a time-critical and computation-intensive task implementation framework by MEC-based MRC system, in which an MR acts as an edge server to provide computation and communication services to multiple SRs. Our aim is to save the energy of robots and prolong the MRC system function time. Accordingly, two resource management and task offloading strategies are proposed to minimize and balance the energy consumption of SRs and MR. In the first scheme, the energy consumption of SRs is minimized and balanced while guaranteeing that the tasks can be accomplished under a time constraint. The second scheme introduces the remaining energy related parameter to enhance the robustness of the system. Simulation results reveal that the MEC-based MRC system can effectively assist the cooperation among robots to accomplish the time-critial and computation-intensive tasks under a strict time constrict.

\section*{Acknowledgement}
This work was supported in part by the National Natural Science Foundation of China (Grant No. 61771429), in part by The Okawa Foundation for Information and Telecommunications, in part by G-7 Scholarship Foundation, in part by the Zhejiang Lab Open Program under Grant 2021LC0AB06, in part by the Academy of Finland under Grant 319759, Zhejiang University City College Scientific Research Foundation (No. JZD18002), in part by ROIS NII Open Collaborative Research 21S0601, and in part by JSPS KAKENHI (Grant No. 18KK0279, 19H04093, 20H00592, and 21H03424).

\section*{Appendix}
\label{Appendix}
\subsection{proof of Lemma 1 }
\label{A}
Note that constraints (16c), (18a) need to be proven convex. As $f(x)$ is a convex function, its perspective function $\frac{t_{k}^{(off)}}{h_k}f(\frac{D_k^{(off)}}{t_{k}^{(off)}})$ is a joint convex function of $D_k^{(off)}$ and $t_{k}^{(off)}$, i.e., constraints (16c), (18a) are convex. Then, all the constraints are convex, thus the problem (P1) turns out to be a convex optimization problem.
\subsection{proof of equation (35) }
\label{B}
Note that $f(x)=N_k(2^{\frac{x}{B_k}}-1)$ and $g(x)=f(x)-xf^{'}(x)$, it has
\begin{equation}\tag{B.1}
	f^{'}(C_kp_k^{(cmp)}h_k)=\frac{N_k\ln2}{B_k}2^{\frac{C_kp_k^{(cmp)}h_k}{B_k}}.
	\label{eqa1}
\end{equation}
Then the inverse function of equation \eqref{eqa1} is
\begin{equation}\tag{B.2}
	f^{'-1}(C_kp_k^{(cmp)}h_k)=B_k\log_2(\frac{BC_kp_k^{(cmp)}h_k}{N_k\ln2}).
	\label{eqa2}
\end{equation}
Next, according to the definitions of $f(x)$ and $g(x)$, the solution for \eqref{eq34} is
\begin{align}\tag{B.3}
	\frac{h_k}{N_k}\gamma_k =& g(f^{'-1}(C_kp_k^{(cmp)}h_k))\notag\\	=&f(f^{'-1}(C_kp_k^{(cmp)}h_k))\notag\\
	&-f^{'-1}(C_kp_k^{(cmp)}h_k)f^{'-1}(f^{'-1}(C_kp_k^{(cmp)}h_k))\notag \\
	=&f(f^{'-1}(C_kp_k^{(cmp)}h_k))\notag\\
	&-f^{'-1}(C_kp_k^{(cmp)}h_k)C_kp_k^{(cmp)}h_k\notag\\
	=&(O_k\ln O_k-O_k+1)\notag,
\end{align}
where $O_k$ is given by \eqref{eq36}.

Furthermore, given $O_k \geq 1$, we proved that $\gamma_k$ is a monotone increasing function. By deriving the first derivatives of \eqref{eq35} with respect to each parameter, it can be written as:
\begin{equation}\nonumber
	\begin{aligned}
		\frac{\partial \gamma_k}{\partial C_k}&=\frac{N_k}{h_k}(\frac{\partial O_k}{\partial C_k}\ln{O_k}), \\
		\frac{\partial \gamma_k}{\partial p_k^{(cmp)}}&=\frac{N_k}{h_k}(\frac{\partial O_k}{\partial p_k^{(cmp)}}\ln{O_k}), \\
		\frac{\partial \gamma_k}{\partial h_k}&=\frac{N_k(O_k-1)}{h_k}.
	\end{aligned}
\end{equation}
According to the properties of logarithm function and the definition of the monotone increasing function, it is easy to conclude that, when $O_k \geq 1$, $\gamma_k$ is a monotone increasing function for parameter $C_k,p_k^{(cmp)},h_k$.
\subsection{proof of Theorem 1.}
\label{C}
1) In order to prove the Theoroem 1, we need the following lemma.
\begin{lemma}\label{lemma2}
	Given $-h_k\gamma_k<0$, the function $g^{'-1}(-h_k\gamma_k)$ is a monotone decreasing function.
\end{lemma}
First, using the definition and property of Lambert function \cite{Lambert}, function $g^{-1}(x)$ is expressed as follow:
\begin{equation}\tag{C.1}
	g^{-1}(x)=\frac{B_k[W_0(\frac{x+N_k}{-N_ke})+1]}{\ln2}.
	\label{c1}
\end{equation}

Because of lambert function is a monotone decreasing function in the domain of definition, thus function $g^{-1}(x)$ is a decreasing function similarly for $x\textless0$.

2) Then, consider the case that $O_k \textless 1$, the threshold function lose its properity of monotone increasing function. And it means the SR has very bad channel state or low local computing energy consumption per bit. Thus, it is reasonable to set $D_k^{(off)}=G_k$, while $t_k^{(off)}$ calculate from \eqref{eq37}. Specifically, when $G_k=0$, it result in $D_k^{(off)}=t_k^{(off)}=0$.

For the case that $O_k \textless 1$, it leads to $D_k^{(off)}>0$, and the time-latency constraint should be active since remaining time can be used for extending offloading duration so as to reduce transmission energy.
\begin{enumerate}
	\item Consider the case where $\gamma_k\textgreater\phi_k^*\textgreater0$, it has $-h_k\gamma_k\textless-h_k\phi_k^*\textless0$. We know $g^{-1}(x)$ is a decreasing function, thus it has
	\begin{equation}\nonumber
		\begin{aligned}
			\frac{D_k^{(off)}}{t_{k}^{(off)*}}=g^{-1}(-h_k\gamma_k)>g^{-1}(-h_k\phi_k^*)\\
			=f^{'-1}(C_kp_k^{(cmp)}h_k)=\frac{D_k^{(off) *}}{t_{k}^{(off)*}}.
		\end{aligned}
	\end{equation}
	Based on the greedy policy, it follows that $D_k^{(off)*}=D_k$.
	\item Consider the case where $\gamma_k=\phi_k^*$, clearly it has
	\begin{equation}\nonumber
		\begin{aligned}
			\frac{D_k^{(off)}}{t_{k}^{(off)*}}=g^{-1}(-h_k\gamma_k)=g^{-1}(-h_k\phi_k^*)\\
			=f^{'-1}(C_kp_k^{(cmp)}h_k)=\frac{D_k^{(off) *}}{t_{k}^{(off)*}}.
		\end{aligned}
	\end{equation}
	It leads to $D_k^{(off)*} \in (G_k,D_k)$.
	\item Consider the case where $\gamma_k\textless\phi_k^*$, it has $-h_k\gamma_k\textgreater-h_k\phi_k^*\textgreater0$, similarly, it has
	\begin{equation}\nonumber
		\begin{aligned}
			\frac{D_k^{(off)}}{t_{k}^{(off)*}}=g^{-1}(-h_k\gamma_k)<g^{-1}(-h_k\phi_k^*)\\
			=f^{'-1}(C_kp_k^{(cmp)}h_k)=\frac{D_k^{(off) *}}{t_{k}^{(off)*}}.
		\end{aligned}
	\end{equation}
	Thus it follows $D_k^{(off)*}=G_k$.
\end{enumerate}
Finally, from \eqref{eq34}, it follows that
\begin{equation}\tag{C.2}
	t_k^{(off)*}=\frac{D_k^{(off)*}}{g^{-1}(-h_k\phi_k^*)}\\
	=\frac{B_k[W_0(\frac{-h_k\gamma+N_k}{-N_ke})+1]}{\ln2}D_k^{(off) *},
	\label{c2}
\end{equation}
where \eqref{c2} is obtained using \eqref{c1}. Ending the proof.

\bibliographystyle{gbt7714-numerical}
\bibliography{myref}

\end{document}